\documentclass{article}
\usepackage{makeidx}
\usepackage[justification=centering]{caption}

\usepackage{graphicx,float}
\usepackage{graphics}
\usepackage{t1enc}
\usepackage{pdfsync}
\usepackage[latin1]{inputenc}

\usepackage[scaled=0.92]{helvet}
\ifdefined\mtpro\usepackage{mtpro2}\fi
\usepackage{mathrsfs,xspace}
\usepackage{amsmath}
\usepackage{verbatim}
\usepackage{authblk} 

\floatstyle{ruled}
\newfloat{Algorithm}{htpb}{alg}

\newcounter{prgline}

\def\..{\,\mathpunct{\ldotp\ldotp}} 

\newcommand{\url}{\cite{myurl}}

\renewcommand{\epsilon}{\varepsilon}
\renewcommand{\phi}{\varphi}

\newtheorem{lemma}{Lemma} 

\newtheorem{theorem}{Theorem}

\newcounter{noqed}
\newcommand{\qed}{ \ifmmode\mbox{
}\fi\rule[-.05em]{.3em}{.7em}\setcounter{noqed}{0}}
\newenvironment{proof}[1][{}]{\noindent{\bf Proof#1.
}\setcounter{noqed}{1}}{\ifnum\value{noqed}=1\qed\fi\par\medskip}

\begin{document} 
\sloppy
\title{Improved space-time tradeoffs for approximate full-text indexing with one edit error\thanks{Most of this work was done when the author was a student at LIAFA, University Paris Diderot - Paris 7. The work was partially supported by the French
 ANR project MAPPI (project number ANR-2010-COSI-004).}}
\author{Djamal Belazzougui}
\affil{Helsinki Institute for Information Technology (HIIT),
Department of Computer Science, University of Helsinki, Finland.}
\bibliographystyle{abbrv}
\maketitle
\begin{abstract}
In this paper we are interested in indexing texts for substring matching queries with one edit error. That is, given a text $T$ of 
$n$ characters over an alphabet of size $\sigma$, we are asked to build a data structure that answers the following query:
find all the $occ$ substrings of the text that are at edit distance at most $1$ from a given string $q$ of length $m$. In this paper we show two new results for this problem. The first result, suitable for an unbounded alphabet, uses $O(n\log^\epsilon n)$ (where $\epsilon$ is any constant such that $0<\epsilon<1$) words of space and answers to queries in time $O(m+occ)$. This improves simultaneously in space and time over the result of Cole et al. The second result, suitable only for a constant alphabet, relies on compressed text indices and comes in two variants: the first variant uses $O(n\log^{\epsilon} n)$ bits of space and answers to queries in time $O(m+occ)$, while the second variant uses $O(n\log\log n)$ bits of space and answers to queries in time $O((m+occ)\log\log n)$. This second result improves on the previously best results for constant alphabets achieved in Lam et al. and Chan et al. 
\end{abstract}
\section{Introduction}
The problem of approximate string matching over texts has been intensively studied. The problem consists in, given a pattern $q$, a text $T$ (the characters of $T$ and $q$ are drawn from the same alphabet of size $\sigma$), and a parameter $k$, to find the starting points of all the substrings of $T$ that are at distance at most $k$ from $q$. There exist many different distances that can be used for this problem. In this paper, we are interested in the edit distance, in which the distance between two strings $x$ and $y$ is defined as the minimal number of edit operations needed to transform $x$ into $y$, where the considered edit operations are deletion of a character, substitution of a character by another and, finally, insertion of a character at some position in the string. Generally two variants of the problem are considered, depending on which of the pattern and the text is considered as fixed. In our case we are interested in the second variant, in which the text is fixed and can thus be processed in advance so as to efficiently answer to queries that consist of a pattern and a parameter $k$. Further, we restrict our interest to the case $k=1$. 
\subsection{Related Work}
The best results we have found in the literature with reference to our problem follow. We only consider the results with worst case space and time bounds. We thus do not consider results like the one in Maa{\ss} and Nowak~\cite{MN05} in which either the query time or the space usage only hold on average on the assumption that the text and/or the patterns are drawn from some random distribution. 
For a general integer alphabet (unbounded alphabet) a result by Amir et al.~\cite{AKLLLR00} further improved by Buchsbaum et al.~\cite{BGW00} has led to $O(n\log^2 n)$~\footnote{In this paper $\log x$ stands for $\lceil\log_2 (\mbox{max}(x,2)\rceil$.} bits of space with query time $O(m\log\log m+occ)$. Later Cole et al.~\cite{CGL04} described an index for an arbitrary number of errors $k$, which for the case $k=1$ uses $O(n\log^2 n)$ bits of space and answers to queries in $O(m+\log n\log\log n+occ)$ time. For the special case of constant-sized alphabets, a series of results culminated with those of Lam et al.~\cite{LSW08}, Chan et al.~\cite{chan2011linear} and Chan et al.~\cite{CLSTW10} with various tradeoffs between the occupied space and query time. 

Belazzougui~\cite{B09} presented a solution to an easier problem: build a data structure for dictionaries so as to support approximate queries with one edit error. In that problem the indexed elements are not text substrings but are, instead, strings coming from a dictionary. Then Belazzougui used his solution to the easier dictionary problem to solve the harder full-text indexing problem. His solution used a brute-force approach: build the proposed dictionary on all sufficiently short text substrings (more precisely, factors with length less than $\log n\log\log n$) and use it to answer queries for sufficiently short query patterns. Queries for longer patterns are answered using the index of Cole et al.~\cite{CGL04}. The main drawback of that solution is that it incurs a large polylogarithmic factor in the space usage (due to the large space needed to index short text substrings). 

By adapting some ideas of Belazzougui~\cite{B09} and combining with two indices described in Chan et al.~\cite{chan2011linear} we are able to remove the additive polylogarithmic term from the query times associated with some of the best previously known results while using the same space (or even less in some cases). When compared with Belazzougui~\cite{B09} our query time is identical but our space usage is much better. 
The reader can refer to Tables~\ref{table:compar_table0} and~\ref{table:compar_table1} for a full comparison between our new results and the previous ones. 

\begin{table}
\centering
\begin{tabular}{|l|l|l|}
  \hline
  Data structure &  Space usage (in bits) & Query time \\
  \hline
  Lam et al.~\cite{LSW08}  & $O(n)$ & $O((m\log\log n+occ)\log^\epsilon n)$\\
  Lam et al.~\cite{LSW08}  & $O(n\log\log n)$ & $O((m\log\log n+occ)\log\log n)$\\
  Lam et al.~\cite{LSW08}  & $O(n\log^\epsilon n)$ & $O(m\log\log n+occ)$\\
  \hline
   Chan et al.~\cite{chan2011linear} & $O(n)$  &  $O(m+\log^{4+\epsilon} n+occ\log^{\epsilon}n)$\\
   Chan et al.~\cite{CLSTW10} & $O(n\log n)$ & $O(m+\log n\log\log n+occ)$ \\
  \hline
  Belazzougui~\cite{B09} & $O(n\log^2 n\log\log n)$ & $O(m+occ)$\\
  \hline
   This paper & $O(n\log\log n)$ & $O((m+occ)\log\log n)$ \\
   This paper & $O(n\log^{\epsilon} n)$ & $O(m+occ)$ \\
  \hline
\end{tabular}
\caption{Comparison of existing solutions for constant alphabet sizes}
\label{table:compar_table0}
\end{table}
\begin{table}
\centering
\begin{tabular}{|l|l|l|}
  \hline
  Data structure &  Space usage (in bits) & Query time \\ 
  \hline
  Buchsbaum et al.~\cite{BGW00} & $O(n\log^2 n)$ & $O(m\log\log n+occ)$\\
  Cole et al.~\cite{CGL04} &  $O(n\log^2 n)$ & $O(m+\log n\log\log n+occ)$\\
  Belazzougui~\cite{B09} & $O(n\log^3 n\log\log n)$ & $O(m+occ)$\\
  This paper & $O(n\log^{1+\epsilon} n)$ & $O(m+occ)$\\ 
  \hline
\end{tabular}
\caption{Comparison of existing solutions for arbitrary alphabets}
\label{table:compar_table1}
\end{table}
As can be seen from the tables, both our indices improve on the state of the art. We should mention that the second result attributed to Lam et al.~\cite{LSW08} in Table 1 is not stated in that paper, but can be easily deduced from the main result in Lam et al.~\cite{LSW08} by using a different time/space tradeoff for the compressed text index implementations. 
The results in Table 1 are all unsuitable for large alphabets as all their query times have a hidden linear dependence on $\sigma$ (which for simplicity is not shown in the table). This means that for very large alphabets of size $\sigma=\Theta(\sqrt n)$  for example, the query time of those algorithms will be unreasonable. By contrast the query times of the algorithms in Table 2 do not have any dependence on the alphabet size. Our result in Table 2 always dominates Cole et al.'s result in both space and time. 


\section{Preliminaries and Outline of the Results}
At the core of our paper is a result for indexing all substrings of the text of length bounded by a given parameter $b$. In particular, we prove the following two theorems: 
\begin{theorem}
\label{constant_fix_theorem}
For any text $T$ of length $n$ characters over an alphabet of constant size, given a parameter $b$, we can build an index of size $O(n(b^\epsilon+\log^\epsilon n)/\epsilon)$ bits  (where $\epsilon$ is any constant such that $0<\epsilon<1$) so that for any given string $q$ of length $m<b$ we can report all of the $occ$ substrings of the text that are at edit distance at most $1$ from $q$ in time $O((m+occ)/\epsilon)$. Alternatively we can build a data structure that occupies $O(n(\log b+\log\log n))$ bits of space and answers to queries in time $O((m+occ)(\log b+\log\log n))$. 
\end{theorem}
Theorem ~\ref{constant_fix_theorem} is obtained by combining ideas from Belazzougui~\cite{B09} with a recent result of Belazzougui et al.~\cite{BBPV10} (weak prefix search) and with known results on compressed text indices. 
The following theorem is an extension of Theorem~\ref{constant_fix_theorem} to unbounded alphabets:  
\begin{theorem}
\label{arb_fix_theorem}
For any text $T$ of length $n$ characters over an alphabet of size $\sigma\leq n$, given a parameter $b$, we can build an index of size $O(n\log n(b^\epsilon+\log^\epsilon n)/\epsilon)$ bits (where $\epsilon$ is any constant such that $0<\epsilon<1$) such that for any given string $q$ of length $m<b$ we can report all of the $occ$ substrings of the text that are at edit distance at most $1$ from $q$ in time $O((m+occ)/\epsilon)$. 
\end{theorem}
Theorem ~\ref{arb_fix_theorem} is obtained by using one crucial idea that was used in Cole et al.'s result~\cite{CGL04} (heavy-light decomposition of the suffix tree) in combination with the weak prefix search result of Belazzougui et al.~\cite{BBPV10} and a recent result of Sadakane~\cite{Sa07} (1D-colored rangeed reporting). 

Theorem~\ref{arb_fix_theorem} can be used for any alphabet size while Theorem~\ref{constant_fix_theorem} holds only for a constant-sized alphabet. Both theorems can only be used for matching strings of bounded length but provide an improvement when used in combination with previous results that are efficient only for long strings.

Theorem~\ref{constant_fix_theorem} gives an immediate improvement for constant-sized alphabets when combined with a result appearing in Chan et al.~\cite{chan2011linear}:
\begin{theorem}
\label{full_constant_theorem}
For any text $T$ of length $n$ over an alphabet of size $\sigma=O(1)$ we can build the following indices, which are both able to return for any query string $q$ of length $m$, the $occ$ occurrences of substrings of $T$ that are at edit distance at most $1$ from $q$:
\begin{itemize}
\item An index that occupies $O(n \log^\epsilon n/\epsilon)$ bits of space and answers to queries in time $O((m+occ)/\epsilon)$ where $\epsilon$ is any constant such that $0<\epsilon<1$. 
\item An index that occupies $O(n\log\log n)$ bits of space and answers to queries in time $O((m+occ)\log\log n)$. 
\end{itemize}
\end{theorem}

Theorem~\ref{arb_fix_theorem} can also be combined with another result that has appeared in Chan et al.~\cite{chan2011linear} to get the following result suitable for arbitrary alphabet sizes:
\begin{theorem}
\label{full_arb_theorem}
For any text $T$ of length $n$ over any integer alphabet of size $\sigma\leq n$ we can build an index that occupies $O(n\log^{1+\epsilon}n/\epsilon)$ bits of space and is able to return for any query string $q$ of length $m$, all the $occ$ occurrences of substrings of $T$ that are at edit distance at most $1$ from $q$ in time $O((m+occ)/\epsilon)$. 
\end{theorem}

Our new methods for proving Theorems~\ref{constant_fix_theorem} and ~\ref{arb_fix_theorem} make use of some ideas introduced in Belazzougui~\cite{B09} combined with tools which were recently proposed in Belazzougui et al.~\cite{BBPV10} and Sadakane~\cite{Sa07} and with compressed text indices proposed in~\cite{FM05,GV05}. In Belazzougui~\cite{B09} a new dictionary for approximate queries with one error was proposed. A naive application of that dictionary to the problem of full-text indexing was also proposed in that paper. However while this leads to the same $O(m+occ)$ query time achieved in this paper, the space usage was too large,  namely $O(n(\log n \log \log n)^2\log\sigma)$ bits of space for alphabet of size $\sigma$. Nonetheless we will borrow some ideas from that paper and use them to prove our main results. 

The paper is organized as follows: we begin with the data structure suitable for constant-sized alphabets (Theorems~\ref{constant_fix_theorem} and~\ref{full_constant_theorem}) in Section~\ref{section:const_alpha} before showing the data structure for large alphabets (Theorems~\ref{arb_fix_theorem} and ~\ref{full_arb_theorem}) in Section~\ref{section:large_alpha}. We  conclude the paper in Section~\ref{section:conclusion}.

\subsection{Model and Notation}
In the remainder, we note by $\overline{x}$ the reverse of the string $x$. That is $\overline{x}$ is the string $x$  written in reverse order. For a given string $s$, we note by $s[i,j]$ or by $s[i..j]$ the substring of $s$ spanning the characters $i$ through $j$. We assume that the reader is familiar with the trie concept~\cite{Fr60,Kn73} and with classical text indexing data structures like suffix trees and suffix arrays (although we provide a brief recall in the next subsection). The model assumed in this paper is the word RAM model with word length $w=\Theta(\log n)$ where $n$ is the size of the considered problem. We further assume that standard arithmetic operations including multiplications can be computed in constant time. We assume that the text $T$ to be indexed is of length $n$ and its alphabet is of size $\sigma<n$. At the end of the paper, we show how to handle the extreme case $\sigma\geq n$. 
\subsection{Basic Definitions}
We now briefly recall some basic (and standard) text indexing data structures that will be extensively used in this paper. 
\subsubsection{Suffix array}
A suffix array~\cite{MM93} (denoted $SA[1..n]$) built on a text $T$ of length $n$ just stores the pointers to the suffixes of $T$ in sorted order (where by order we mean the usual lexicographic order defined for strings). Clearly a suffix array occupies $n\log n$ bits. A suffix array example is illustrated in Figure~\ref{pic:suffix_array}
\begin{figure}[htb] 
\centering\includegraphics[scale=1]{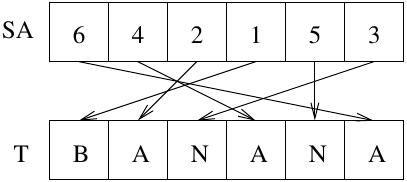} 
\caption{Suffix array for the text $banana$} 
\label{pic:suffix_array} 
\end{figure}
\subsubsection{Suffix tree}
A suffix tree~\cite{Wr73,Mc76} is a \emph{compacted} trie built on the suffixes of a text $T$ appended with $\#$, a special character outside of the original alphabet and smaller than all the characters of the original alphabet. 
A suffix tree has the following properties:
\begin{itemize}
\item Every suffix of $T$ is associated with a leaf in the tree.
\item A factor $p$ of $T$ is associated with an internal node in the tree iff there exist two characters $a$ and $b$ such that $pa$ and $pb$ are also factors of $T$. Each internal has thus at least two children, one associated with a factor that starts with $pa$ and another associated with a factor that starts with $pb$. 
\item The subtree rooted at any internal node associated with a prefix $p$ contains (in its leaves) all the suffixes of $T\#$ which have $p$ as a prefix. 
\item Suppose that an edge connects an internal node $x$ associated with a factor $p$ to a node $y$ (which could be a leaf or an internal node) associated with a string $s=pcs'$ (which could either be a suffix or another factor) where $c$ is a character and $s'$ is string. Then the character $c$ will be called the \emph{label} of $y$ and $s'$ will be called the compacted path of $y$. 
\end{itemize}
 The essential property of a suffix tree is that it can be implemented in such a way
that it occupies $O(n)$ pointers (that is $O(n\log n)$ bits) in addition to the text and that given any factor $p$ of $T$ it is possible to find all the suffixes of $T$ which are prefixed by $p$ in $O(|p|)$ time. A suffix tree can also be augmented in several ways so as to support many other operations, but in this paper we will use very few of them. A suffix tree example is illustrated in Figure~\ref{pic:suf_tree}.

For a more detailed description of the suffix array or the suffix tree, the reader can refer to any book on text indexing algorithms~\cite{Gu97,CR03}.
\begin{figure}[htb] 
\centering\includegraphics[scale=1]{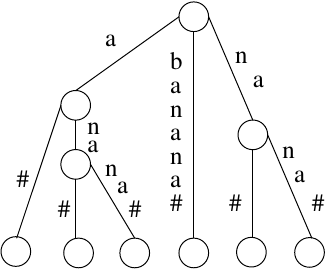} 
\caption{Suffix tree for the text $banana$} 
\label{pic:suf_tree} 
\end{figure}

\section{Solution for Constant-Sized Alphabets}
\label{section:const_alpha}
In this section we give a proof of Theorems~\ref{constant_fix_theorem} and~\ref{full_constant_theorem}. 
Theorem~\ref{full_constant_theorem} is proved in Section~\ref{section:arb_const_query}. This theorem uses the data structure of Theorem~\ref{constant_fix_theorem}, which is described in Section~\ref{section:const_fix_data_struct}. Then  we describe how the queries are executed on the data structure in Section~\ref{section:const_fix_query}. Finally, in Section~\ref{subsec:dupl_occ} we show how we deal with duplicate occurrences. 
\subsection{Solution for Arbitrary Pattern Length}
\label{section:arb_const_query}
We use the following lemma proved by Chan et al.~\cite[Section 3.2]{chan2011linear}.
\begin{lemma}~\cite[Section 3.2]{chan2011linear}
\label{CLSTW06b0_fix_lemma}
For any text $T$ of length $n$ characters over an alphabet of constant-size we can build an index of size $O(n)$ bits so that we can report all of the $occ$ substrings of the text which are at edit distance $1$ from any pattern $q$ of length $m\geq \log^4 n\log\log n$ in time $O(m+occ)$. 
\end{lemma}

The solution for Theorem~\ref{full_constant_theorem} is easily obtained by combining Theorem~\ref{constant_fix_theorem} with Lemma~\ref{CLSTW06b0_fix_lemma} in the following way: we first build the index of Chan et al.~\cite{chan2011linear} whose query time is upper bounded by $O(m+occ)$ whenever $m\geq \log^4 n\log\log n$ and whose space usage is $O(n)$ bits. Then we build the data structure of Theorem~\ref{constant_fix_theorem} in which we set $b=\log n^4\log\log n$ and $\epsilon=\delta/5$ where $\delta$ is any constant that satisfies $0<\delta<1$. In the case where we have a string of length less than $b$, we use the index of Theorem~\ref{constant_fix_theorem} to answer the query in time $O((m+occ)/\epsilon)=O((m+occ)/\delta)$ when using the first variant, or in time $O((m+occ)\log\log n)$ when using the second variant. In the case where we have a string of length at least $b=\log^4 n\log\log n$, we use the index of Chan et al.~\cite{chan2011linear} answering to queries in time $O(m+occ)$. The space is thus dominated by our index, which uses either $O(n(\log^4 n\log\log n)^{\epsilon}/\epsilon)=O(n\log^\delta n/\delta)$ or $O(n\log(\log^4 n\log\log n))=O(n\log\log n)$ bits of space.

\subsection{Data Structure for Short Patterns}
\label{section:const_fix_data_struct}
Our data structure for short patterns relies on a central idea used in Belazzougui~\cite{B09}. That paper was concerned with building an approximate dictionary that had to support searching queries that tolerate one edit error. The idea for obtaining the result was that of using a hash-based dictionary combined with a trie and a reverse trie, so that finding all the strings in the dictionary which are at distance $1$ from a pattern $q$ of length $m$ can be done in constant time (amortized) per pattern character and constant time per reported occurrence. This constant time derives from two facts:
\begin{enumerate}
\item If the dictionary uses some suitable perfect hash function $H$, then after we have done a preprocessing step on $q$ in $O(m)$ time, the computation of $H(p)$ takes constant time for any string $p$ at distance $1$ from $q$. 
\item Using the trie and reverse trie, the matching of any candidate string $p$ at distance $1$ from $q$ can be verified in constant time (this idea has frequently been used before, for example by Brodal and G{\c{a}}sieniec~\cite{BG96}). 
\end{enumerate}
In our case we will use different techniques from the ones used in Belazzougui~\cite{B09}. As we are searching in a text rather than a dictionary, we will be looking for suffixes prefixed by some string $p$ instead of finding exact matching entries in a dictionary. For that purpose we will replace the hash-based dictionary with a weak prefix search data structure~\cite{BBPV10},  which will allow us to find a range of candidate suffixes prefixed by any given string $p$ at edit distance $1$ from $q$. For checking the candidate suffixes, we will replace the trie and reverse trie with compressed suffix arrays built on the text and the reverse of the text. Both the weak prefix search and the compressed suffix arrays have the advantage of achieving interesting tradeoffs between space and query time. 

We now describe in greater detail the data structure we use to match patterns of bounded length over small alphabets (Theorem~\ref{constant_fix_theorem}). This data structure uses the following components:

\begin{enumerate}
\item A suffix array $SA$ built on the text $T$. 
\item A suffix tree $S$ built on the text $T$. In each node of the suffix tree representing a factor $p$ of $T$, we store the range of suffixes which start with $p$. That is we store a range $[i,j]$ such that any suffix starts with $p$ iff its rank $k$ in lexicographic order is included in $[i,j]$. 
\item A reverse suffix tree $\overline{S}$ built on the text $\overline{T}$, the reverse of the text $T$ (we could call $\overline{S}$ a prefix tree as it actually stores prefixes of $T$). In each node of $\overline{S}$ representing a factor $p$ we store the range of suffixes of $\overline{T}$ which start with $p$. That is we store a range $[i,j]$ such that any suffix of $\overline{T}$ starts with $p$ iff its rank $k$ in lexicographic order among all the suffixes of $\overline{T}$ is included in $[i,j]$. 
\item A table $SA^{-1}[1..n]$. This table stores for each suffix $T[i..n]$ for all $1\leq i\leq n$, the rank of the suffix $T[i..n]$ in lexicographic order among all the suffixes of $T$. 
\item A table $PA^{-1}[1..n]$. This table stores for each prefix $T[1..i]$ the rank of the reverse of prefix $T[1..i]$ in lexicographic order among the reverses of all prefixes of $T$.  
\item A polynomial hash function $H$~\cite{KR87} parameterized with a prime $P>n^4$ and an integer $r\in[1,P-1]$ (a seed). For a string $x$ we have $H(x)=(x[1]\cdot r+x[2]\cdot r^2+...+x[|x|]\cdot r^{|x|})\bmod P$. The details of the construction are described below. The hash function essentially uses just $O(\log n)$ bits of space to store the numbers $P$ and $r$. 
\item A weak prefix search data structure (which we denote by $W_0$) built on the set $U$, the set of substrings (factors) of $T$ of fixed length $b$ characters (to which we add $b-1$ artificial factors obtained by appending $\#^{b-i}$ to every suffix of $T$ of length $i<b$ ). Note that $|U|\leq n$. This data structure, which is described in Belazzougui et al.~\cite{BBPV10}, comes in two variants. The first one uses $O(n(b^\epsilon+\log\log\sigma))$ bits of space for any constant $0<\epsilon<1$ and answers to queries in $O(1)$ time.  The second one uses space $O(n(\log b+\log\log\sigma))$ and has query time $O(\log b)$. We note the query time of the weak prefix search data structure by $t_{W_0}$. The details are described below. 
\item Finally a prefix-sum data structure $V_0$ built on top of an array $C_0[1..|U|]$ which stores for every $p\in U$ sorted in lexicographic order, the number of suffixes of $T$ prefixed by $p$ (for each of the artificial strings this number is set to one). This prefix-sum data structure uses $O(|U|)=O(n)$ bits of space and answers in constant time to the following queries: given an index $i$ return the sum $\sum_{1\leq j\leq i} C_0[1..|U|]$. The details are given below. 
\end{enumerate}
Note that the total space usage is dominated by the text indexing data structures (the suffix and prefix trees and the structures $SA$,$SA^{-1}$,$PA$,$PA^{-1}$) which occupy $O(n\log n)$ bits of space. These data structures are illustrated in Figures~\ref{pic:text_idx_DS} and ~\ref{pic:text_idx_DS2}. 

\begin{figure}[htb] 
\centering\includegraphics[scale=0.8]{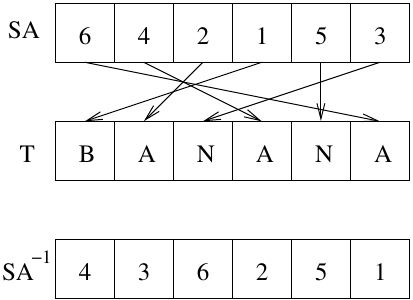} 
\includegraphics[scale=0.8]{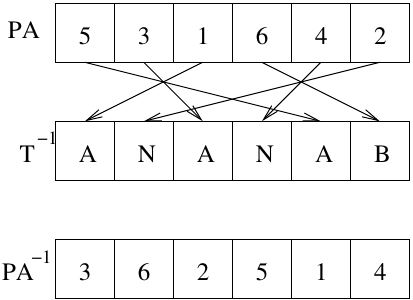} 

\caption[justification=centering]{Data structures $SA$,$SA^{-1}$,$PA$ and $PA^{-1}$ for the string $banana$} 

\label{pic:text_idx_DS} 
\end{figure}

\begin{figure}[htb] 
\centering\includegraphics[scale=1]{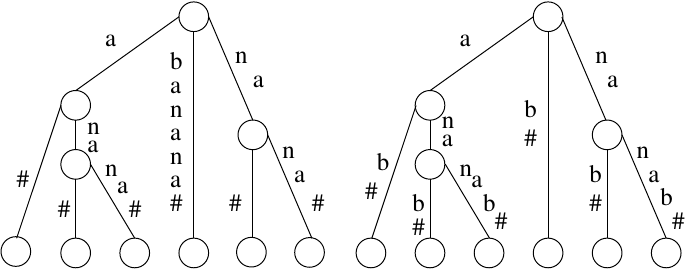} 
\caption{The suffix tree (on the left) and the prefix tree (on the right) for the string $banana$} 
\label{pic:text_idx_DS2} 
\end{figure}

We now describe in detail the results from the literature that will be used to implement the data structures above. 

\subsubsection{Text indexing data structures}
The only operation we need to do on the prefix tree is, for a given pattern $q$, to determine for each prefix $p$ of $q$ the range of all prefixes of $T$ which have $p$ as a suffix. Similarly for the suffix tree we only need to know for each suffix $s$ of $q$ the range of suffixes which are prefixed by $s$. 
The classical representations for our text indexing data structures ($SA$,$SA^{-1}$,$PA$, suffix and prefix trees) all occupy $O(n\log n)$ bits of space. However in our case, we need to use less than the $O(n\log n)$ bits needed by the classical representations. We will thus make use of compressed representations of the text indexing data structures~\cite{FM05,GV05}. In particular we use the following results: 
\begin{enumerate}
\item For every prefix $p_i=q[1,i]$ of $q$ of length $i$ determine the range $[pl_i,pr_i]$ of prefixes of $T$ which are suffixed by $p_i$. This can be  accomplished incrementally in $O(m)$ time by following the suffix links~\footnote{A suffix link connects a suffix tree node associated with a factor $cp$ (with $c$ being a character) to the suffix tree node associated with the factor $p$.} in the prefix tree $\overline{S}$. That is, deducing the range corresponding to the prefix of $q$ of length $i$ from the range of the prefix of $q$ of length $i+1$ in $O(1)$ time (following a suffix link at each step takes $O(1)$ time). In the context of compressed data structures, this can be accomplished using the backward search on the compressed representation of the prefix array $PA$~\cite{FM05} still in time $O(m)$ and representing $PA$ in $O(n)$ bits only (assuming a constant alphabet)~\footnote{$PA$ can be represented in compressed suffix array representation of~\cite{FM05}, since $PA$ is actually the suffix array of the reverse of the text.}. In this case the range corresponding to the prefix of $q$ of length $i+1$ (that is, $p_{i+1}$) is deduced from the range corresponding to the prefix of $q$ of length $i$ (that is, $p_i$). 

\item For every suffix $s_i=q[m-i+1,m]$ of $q$ of length $i$ determine the range $[sl_i,sr_i]$ of suffixes of $T$ which are prefixed by $s_i$. This can be done in a similar way in total $O(m)$ time by either following suffix links in a standard representation of the suffix tree $S$ or by backward search~\cite{FM05} in a compressed representation of the suffix array $SA$. The compressed representation occupies $O(n)$ bits only.  
\item For any $i$ we need to have a fast access to $SA[i]$,$SA^{-1}[i]$,$PA[i]$,$PA^{-1}[i]$. In case those four tables are represented  explicitly in $O(n\log n)$ bits of space, the access time is trivially $O(1)$. However in the context of compressed representation, we need to use less than $O(n\log n)$ bits of space and still be able to have fast access to the arrays. 
\end{enumerate}
The first two results can be summarized with the following lemma:
\begin{lemma}~\cite{FM05}
Given a text $T$ of length $n$ over a constant-sized alphabet, we can build a data structure with $O(n)$ bits of space such that given a pattern $q$ of length $m$, we can in $O(m)$ time determine:
\begin{itemize}
\item the range of suffixes of $T$ (sorted in lexicographic order) prefixed by $s_i$ for all suffixes $s_i$ of $q$ of length $i\in[1..m]$.
\item the range of prefixes of $T$ (sorted in reverse lexicographic order) suffixed by $p_i$ for all prefixes $p_i$ of $q$ of length $i\in[1..m]$.
\end{itemize}
If the alphabet is non constant, then we can obtain the same results using $O(n\log\sigma)$ bits of space~\cite{BN11}. 
\end{lemma}
The third needed text indexing result is summarized with the following lemma:
\begin{lemma}~\cite{GV05,Rao02,LSW05}
\label{lemma:compr_suf_array}
Assuming a constant alphabet size, we can compress the arrays $PA$,$SA$,$PA^{-1}$ and $SA^{-1}$ with the following tradeoffs : 
\begin{itemize}
\item space $O(n\log\log n)$ bits with access time $t_{SA}=O(\log\log n)$ time. 
\item space $O(n\log^\epsilon n/\epsilon)$ bits with access time $t_{SA}=O(1/\epsilon)$ time, where $\epsilon$ is any constant
such that $0<\epsilon<1$. 
\end{itemize}
\end{lemma}

\subsubsection{Weak prefix search data structure}
A weak prefix search data structure built on a set of strings $U$ (sorted by increasing lexicographic order) permits, given a prefix $p$ of any element in $U$, to return the range of elements of $U$ prefixed by $p$. If given an element which is not prefix of any element in $U$, it returns an arbitrary range. We will use the following result:
\begin{lemma}~\cite{BBPV10}
\label{lemma:weak_pref_search}
Given a set of $n$ strings of fixed length $b$ each over the alphabet $[1..\sigma]$, we can build a weak prefix search data structure with the following time/space tradeoffs:
\begin{itemize}
\item Query time $t_W=O(c)$ with a data structure which uses $O(nc(b^{1/c}(\log b+\log\log\sigma)))$ bits of space\footnote{Actually the result in Belazzougui et al.~\cite{BBPV10} states a space usage $O(nb^{1/c}\log b)$ but assumes a constant alphabet size. However, it is easy to see that the same data structure just works for arbitrary $\sigma$ in which case it uses $O(n(b^{1/c}(\log b+\log\log\sigma)))$ bits of space.} for any integer constant $c>1$. 
\item Query time $t_W=O(\log b)$ with a data structure which uses $O(n(\log b+\log\log\sigma))$ bits of space. 
\end{itemize}
\end{lemma}
The weak prefix search data structure of the lemma above needs to use a perfect hash function $H$ and assumes that after preprocessing a query string $p$, the computation of $H(p[1,i])$ for any $i$ takes constant time. This is essential to achieve $t_W$ query time, since a query involves either $\Theta(c)$ or $\Theta(\log b)$ computations of the hash function on prefixes of $p$. 

In our case, the weak prefix search data structure will be built on $U$, the set of factors of $T'=T\#^{b-1}$ of fixed length $b$, where $\#$ is a special character lexicographically smaller than all the characters which appear in $T$. The only reason we use $T'$ instead of $T$ in the weak prefix search is to ensure that the last suffixes of length less than $b$ are all present in the weak prefix (a suffix $s$ of length $i<b$ will be stored in the weak prefix search as the string $s\#^{b-i}$).

\subsubsection{Hash function}
As described above we will use a polynomial hash function $H$~\cite{KR87} parameterized with a prime number $P$ and a seed $r$. The parameter $P$ is fixed but the seed $r$ is chosen randomly. Our goal is to build a hash function $H$ such that all the hash values of the substrings of $T$ used by the weak prefix search are all distinct (that is, $H$ is a perfect hash function for the considered set if substrings). For a randomly chosen $r$ this is the case with high probability. If it is not the case, we randomly choose a new $r$ and repeat the construction until all the needed substrings of $T$ are mapped to distinct hash values. 
We notice that the hash function $H$ can be evaluated on any string $s$ in deterministic linear time in the length of $s$. However the time to find a suitable $r$ is randomized only, as we may in the worst case do many trials before finding a suitable $r$ (although the number of trials is on average only $1+o(1)$ as each trial succeeds with high probability). 

\subsubsection{Prefix-sum data structure}
A succinct prefix-sum data structure is a data structure that permits the succinct encoding of an array $A[1..n]$ of integers of total sum $D$ in space $n(2+\lceil\log(D/n)\rceil)$ bits, so that the sum $\sum_{1\leq j\leq i}A[j]$ for any $i$ can be computed in constant time. This can be obtained by combining fast indexed bitvector implementations ~\cite{J89,Mu96,CM96} with Elias-Fano coding~\cite{EliESRCASF,FanNBRISM}. 
\subsection{Queries}
\label{section:const_fix_query}
\subsubsection{Preprocessing}
To make a query on our full-text index for a string $q$ of length $m$, we will proceed in a preprocessing step which takes $O(m)$ time. The preprocessing consists in the following phases: 
\begin{enumerate}
\item We fill two arrays $L[0..m]$ and $R[1..m+1]$ by using Lemma~\ref{lemma:compr_suf_array} on the string $q$. More precisely $L[i]$ stores the range of prefixes suffixed by $q[1..i]$ and $R[i]$ stores the range of suffixes prefixed by $q[i,m]$ (we naturally associate the range $L[0]=R[m+1]=[1,n]$ with the empty strings $q[1..0]$ and $q[m+1..m]$ which, are respectively suffix and prefix of any other string). This step takes time $O(m)$ as stated in Lemma~\ref{lemma:compr_suf_array}. 
\item We precompute an array which stores all the values of $r^{i}$ for all $0\leq i\leq m$.
\item We precompute all the values $H(q[1,i])$ for all $1\leq i\leq m$. That is all the hash values for all the prefixes of $q$. This can easily be done incrementally as we have $H(q[1,1])=q[1]\cdot r$ and then $H(q[1,i+1])=H(q[1,i])+q[i+1]\cdot r^{i+1}$ for all $1\leq i<m$. 
\item We precompute all the values $H(q[m-i+1,m])$ for all $1\leq i\leq m$. That is all the hash values for all the suffixes of $q$. This can also easily be done incrementally as we have $H(q[m,m])=q[m]\cdot r$ and then $H(q[i,m])=(H(q[i+1,m])+q[i])\cdot r$ for all $1\leq i<m$. 
\end{enumerate}

\subsubsection{Hash function computation}
We now describe some useful properties of the hash function $H$ which will be of interest for queries. 
An interesting property of the hash function $H$ is that after the precomputation phase, computing $H(p)$ for any $p$ at edit distance $1$ from $q$ takes constant time:
\begin{enumerate}
\item Deletion at position $i$: computing the hash value of $p=q[1,i-1]q[i+1,m]$ (note that $q[1,i]$ is defined as the empty string when $i=0$) is done by the formula $H(p)=H(q[1,i-1])+H(q[i+1,m])\cdot r^{i-1}$ (note that $H(q[1,i-1])=0$ if $i=1$). 
\item Substitution at position $i$: computing the hash value of $p=q[1,i-1]cq[i+1,m]$ is done by the formula $H(p)=H(q[1,i-1])+(c+H(q[i+1,m]))\cdot r^i$. 
\item Insertion after position $i$: computing the hash value of $p=q[1,i]cq[i+1,m]$~\footnote{Note that insertion before position $1$ is equivalent to insertion after position $0$ in which case $q[1,i]$ will be the empty string.} is done by the formula $H(p)=H(q[1,i])+(c+H(q[i+1,m]))\cdot r^{i+1}$.
\end{enumerate}
It can easily be seen that the computation of $H(p)$ takes constant time in each case. The reason is that the three values involved in each computation have all been obtained in the precomputation phase. 

Moreover, computing $H(p')$ for any prefix $p'$ of a string $p$ at edit distance $1$ from $q$ also takes constant time:
\begin{itemize}
\item The hash value for a prefix $p'$ of length $j$ of a string $p$ obtained by deletion at position $i$ in $q$ can be obtained by $H(p')=H(p)-H(q[j+2,m])\cdot r^j$ (note that $H(q[j+2,m]=0$ if $j+2>m$) whenever $j\geq i$ or $H(p')=H(q[1,j])$ otherwise. 
\item The hash value for a prefix $p'$ of length $j$ of a string $p$ obtained by substitution at position $i$ in $q$ can be obtained by $H(p')=H(p)-H(q[j+1,m])\cdot r^j$ (note that $H(q[j+1,m]=0$ if $j+1>m$) whenever $j\geq i$ or $H(p')=H(q[1,j])$ otherwise.
\item The hash value for a prefix $p'$ of length $j$ of a string $p$ obtained by insertion after position $i$ in $q$ can be obtained by $H(p')=H(p)-H(q[j,m])\cdot r^j$ (note that $H(q[j,m]=0$ if $j>m$) whenever $j>i$ or $H(p')=H(q[1,j])$ otherwise. 
\end{itemize}
\subsubsection{Checking occurrences}
Suppose we have found a potential occurrence of a matching substring of the text obtained by one deletion, one insertion or one substitution. There exists a standard way to check for the validity of the matching (this has been used several times before, for example in Chan et al.~\cite{CLSTW10}) using the arrays $PA^{-1}$ and $SA^{-1}$. Suppose we have found a potential occurrence of a string $p$ obtainable by deletion of the character at position $i$ in the query string $q$. In this case  we have $p=q[1,i-1]q[i+1,m]$. Moreover, suppose we have found for $p$ a potentially matching location $j$ in the text. Then checking whether this matching location is correct is a matter of just checking that $PA^{-1}[j+i-2]\in L[i-1]$ and, $SA^{-1}[j+i-1]\in R[i+1]$. That is, checking whether $T[j..j+m-2]=p=q[1,i-1]q[i+1,m]$ is just a matter of checking that $T[j..j+i-2]=q[1..i-1]$, and $T[j+i-1..j+m-2]=q[i+1,m]$ which amounts to checking that the prefix of $T$ ending at $T[j+i-2]$ is suffixed by $q[1,i-1]$, and the suffix of $T$ starting at $T[j+i-1]$ is prefixed by $q[i+1,m]$. Those two conditions are equivalent to checking that $PA^{-1}[j+i-2]\in L[i-1]$, and $SA^{-1}[j+i-1]\in R[i+1]$ respectively. Checking a matching location for an insertion or a substitution can be done similarly. For checking a matching at position $j$ in the text of a pattern obtainable by insertion of a character $c$ after position $i$, it suffices to check that $PA^{-1}[j+i-1]\in L[i]$, $T[j+i]=c$ and finally $SA^{-1}[j+i+1]\in R[i+1]$. Similarly checking for a matching of a string obtainable by a substitution at position $i$ can be done by checking that $PA^{-1}[j+i-2]\in L[i-1]$, $T[j+i-1]=c$ and finally $SA^{-1}[j+i]\in R[i+1]$. 
We thus have the following lemma which will be used as a central component for query implementation:
\begin{lemma}
\label{check_occ_lemma}
Given any pattern $q$ for which the arrays $L$ and $R$ have been precomputed, we can, for any string $p$ at distance $1$ from $q$ (where $p$ is described using $O(1)$ words that store the edit operation, the edit position and a character in case the operation is an insertion or a substitution) and a location $\ell$, check whether $p$ occurs at location $\ell$ in the text by doing a constant number of probes to the text and the arrays $SA^{-1}$, $PA^{-1}$  and thus in total time $O(t_{SA})$. 
\end{lemma}

\subsubsection{Query algorithm}
\label{section:const_query_algo1}
We now describe how queries for a given string $q$ of length $m$ are implemented. Recall that we are dealing with an alphabet of constant size $\sigma$. This means that the number of strings at distance $1$ from a given $q$ is $O(m\sigma)=O(m)$, because an edit operation is specified by a position and a character (except for deletions which are specified only by a position) and thus, we can have at most $O(m\sigma)$ different combinations. 

Our algorithm will simply check exhaustively for matching in the text of every string $p$ which can be obtained by one insertion, one deletion or one substitution in the string $q$. 
Each time we check for a string $p$ we also report the location of all occurrences in which it matches. The matching for a given string $p$ proceeds in the following way:
\begin{itemize}
\item Do a weak prefix search on $W_0$ for the string $p$ which takes either constant time or $O(\log b)$ time depending on the implementation used (see Lemma~\ref{lemma:weak_pref_search})~\footnote{Note that the query time bound uses the essential fact that the function $H$ can be computed on any prefix $p'$ of $p$ in constant time.}. The result of this weak prefix search is a range $[\ell_0,r_0]$ of elements in $U$ which are potentially prefixed by $p$. 

\item Using the prefix-sum data structure $V_0$, compute the range $[\ell_1,r_1]$ of suffixes of $T$ potentially prefixed by $p$. This range is given by $\ell_1=\sum_{1\leq t<l_0}V_0[t]$ and $r_1=\sum_{1\leq t\leq r_0}V_0[t]$ and its computation takes constant time. 
\item Do a lookup for $j_1=SA[\ell_1]$ in the suffix array. This takes either time $O(\log\log n)$ or $O(1)$ depending on the suffix array implementation (see Lemma~\ref{lemma:compr_suf_array}). 
\item Finally check that there is a match in the location $j_1$ in the text (this is done differently depending on whether we are dealing with an insertion, a substitution or a deletion). This checking, which is done with the help of Lemma~\ref{check_occ_lemma}, needs to do one access to $PA^{-1}$ and one access to $SA^{-1}$ and thus takes either $O(\log\log n)$ or $O(1)$ time depending on the implementation. If the match is correct we report the position $j_1$ and additionally report all the remaining matching locations which are at positions $SA[\ell_1+1],SA[\ell_1+2],,,SA[r_1-1],SA[r_1]$ (by querying the compressed suffix array using Lemma~\ref{lemma:compr_suf_array}). Otherwise we return an empty set. 
\end{itemize}
For proving the correctness of the query, we first prove the following lemma:
\begin{lemma}
Given a string $p$ of length at most $b$ such that $p$ is a prefix of at least one suffix of $T$, with the help of $W_0$ and $V_0$, we can find the interval of suffixes prefixed by $p$ in time  $O(t_{W_0})$. Further, given any $p$ we can check whether it prefixes some suffix of $T$ in time $O(t_{W_0}+t_{SA})$ and if not return an empty set. 
\end{lemma}
\begin{proof}
We start with the first assertion. If $p$ is prefix of some suffix $s\in T$ then it will also be prefix of some element in $U$. This is trivially the case if $|s|\geq b$ and this is also the case if $|s|<b$ as we are storing in $U$ the string $s\#^{b-|s|}$ which is necessarily prefixed by $p$ as well. Now we prove that the returned interval $[l_1,r_1]$ is the right interval of suffixes prefixed by $p$. First notice that the weak prefix search by definition returns the interval $[l_0,r_0]$ of elements of $U$ which are prefixed by $p$. Now, we can easily prove that $l_1=\sum_{1\leq t<l_0}V_0[t]$ is exactly the number of suffixes that are lexicographically smaller than $p$. This is the case as we know that the sum $\sum_{1\leq t<l_0}V_0[t]$ includes exactly the following:
\begin{enumerate}
\item All suffixes of length less than $b$ which are lexicographically smaller than $p$ and for which an artificial element was inserted in $U$. 
\item All the suffixes of length at least $b$ whose prefixes of length $b$ are lexicographically smaller than $p$. 
\end{enumerate}
On the other hand we can prove that $r_1-l_1$ gives exactly the number of suffixes prefixed by $p$. That is $r_1-l_1+1=\sum_{l_0\leq t\leq r_0}V_0[t]$ which gives the number of suffixes of length at least $b$ prefixed by elements of $U$ prefixed by $p$ in addition to the suffixes of length less than $b$ prefixed by $p$ and for which a corresponding artificial element has been stored in $U$. Now that the first assertion of the lemma has been proved we turn our attention to the second assertion. This is immediate: given a string $p$ which does not prefix any suffix, we know by Lemma~\ref{check_occ_lemma} that the checking will fail for any suffix of $T$ and thus fail for the suffix at position $j$ in the last step which thus returns an empty set. 
\qed
\end{proof}

The following lemma summarizes the query for a prefix $p$ at distance one from $q$: 
\begin{lemma}
\label{modif_pattern_query}
Given any pattern $q$ for which the arrays $L$ and $R$ have been computed, we can for any string $p$ at edit distance $1$ from $q$ (where $p$ is described using $O(1)$ words that store the edit operation, the edit position and a character in case the operation is an insertion or a substitution) search for all the $occ$ suffixes prefixed by $p$ in time $O(t_{W_0}+(occ+1)t_{SA})$. 
\end{lemma}
\begin{proof}
We first prove the correctness of the operations as described above and then prove the time bound. For that we examine the two possibilities: 
\begin{itemize}
\item $p$ is not prefix of any suffix in $T$, in which case the query should return an empty set. It is easy to prove the equivalent implication: if the data structure returns a non empty set, then there exists at least some suffix of $T$ prefixed by $p$. For the data structure to return a non empty set, the checking using Lemma~\ref{check_occ_lemma} must return true for the location $j$ in the text and for this checking to return true, $q$ must be prefix of the suffix starting at position $j$ in the text. 
\item $p$ is prefix of some suffix of $T$, in which case the query must return all those suffixes. Note that by definition the weak prefix search $W_0$ will return the right range of elements of $U$ which are prefixed by $p$. 
The justification for this is that a single match implies that there exists at least one suffix prefixed by $p$, which implies the weak prefix search $W_0$ returns the correct range of factors in $U$ and thus the bit-vector $V_0$ returns the correct range of suffixes $SA[l_1]$,$SA[l_1+1],\ldots,SA[r_1-1],SA[r_1]$ which must also have $p$ as a prefix. 
\end{itemize}
We now prove the time bound. In the case that $p$ is not prefix of any suffix in $T$, the query time is clearly $O(t_{W_0}+t_{SA})$ as we are doing one query on $W_0$ (which takes $O(t_{W_0})$ time), one query of the prefix-sum data structure $V_0$ in constant time and finally the checking using Lemma~\ref{check_occ_lemma} which takes $O(t_{SA})$ time. In case $p$ is prefix of some suffixes, then the first step for checking that the set is non empty also takes constant time and reporting each occurrence takes additional $t_{SA}$ time per occurrence. 
\qed
\end{proof}

\subsection{Duplicated occurrences}
\label{subsec:dupl_occ}
A final, important detail is to avoid reporting the same occurrence more than once. This can happen in a few cases if insufficient care is taken in the query algorithm of Section~\ref{section:const_query_algo1}. As an example, take the case of the string $q=a^mb^m$; deleting one character at any of the $m$ first positions of $q$ would result in the same string $a^{m-1}b^m$. The same problem occurs if we delete one of the last $m$ characters of $q$, resulting in the same string $a^mb^{m-1}$. 
In order to avoid needing to check the same string more than once we can use the following simple rules: 
\begin{enumerate}
\item For deletions, we avoid testing for the deletion of a character $q[i]$ if $q[i-1]=q[i]$ as this has the same effect as deleting character $q[i-1]$. 
\item For substitutions we avoid testing for the replacement of the character $q[i]$ by the same $q[i]$, as this would result in the same string $q$. 
\item For insertions, we avoid inserting the character $c=q[i]$ between positions $i$ and $i+1$ as it would have the same effect as inserting character $c$ between positions $i-1$ and $i$. 
\end{enumerate}
It is easy to check that the three rules above will avoid generating the same candidate strings for each of the three edit operations~\footnote{This is evident for substitutions. It is also true for deletion, since only the last character in a run of equal characters is deleted. For insertions, the only problematic case is when inserting the same character in a run of equal characters of length at least $1$, and this case is  avoided since we only insert such a character at the end of the run.}. Note however that a string obtained by two different operations could still generate the same locations in the text. As an example the string $q=a^mb^m$ can generate the string $p_1=a^{m-1}b^m$ by a delete operation and the string $p_2=a^{m-1}b^{m+1}$ by a substitution. Note that all occurrences of $p_2$ are also occurrences of $p_1$. 
There exist a few different ways to fix this problem. The first is to just not do anything if the user can tolerate that an occurrence is reported a constant number of times. Indeed each occurrence can only be reported at most three times (once for each edit operation) as all occurrences generated by each edit operation must be distinct. 

If we require that each occurrence is only reported once, then we can just use a bitvector of $n$ bits to track the already reported occurrences. Before starting to answer to queries, we initially set all the positions in the bitvector to $0$. Then when executing a query we only report an occurrence if the corresponding position in the bitvector is unmarked and then mark the position. Finally after the query has finished, we reset all the positions that were marked during the query so that the bits in the bitvector are again all zero before the next query is executed. 

\section{Solution for Large Alphabets}
\label{section:large_alpha}
The time bound of Theorem~\ref{constant_fix_theorem} has a linear dependence on the alphabet size as the query time is actually $O(\sigma m+occ)$. This query time is not reasonable in the case where $\sigma$ is non constant. In this section we show a solution which has no dependence on alphabet size. In order to achieve this solution we combine the result of Chan et al.~\cite{chan2011linear} with an improved version of Chan et al.~\cite{CLSTW10} that works only for query strings of length bounded by a parameter $b$. This is shown in Section~\ref{section:full_arb_solution}. We then show how we improve the  solution in Chan et al.~\cite{CLSTW10} in Sections~\ref{section:arb_high_level_desc},~\ref{section:DS_impl} and~\ref{section:gen_weight_decomp}. 
This improvement is an extension of the data structures described in the previous section, but with two major differences: it does not use compressed variants of the text indexing data structures and its query time is independent of the alphabet size. Before describing the details of the used data structures, we recall in Section~\ref{section:arb_tools} a few definitions and data structures which will be used in our construction. 

\subsection{Solution for Arbitrary Pattern Length}
\label{section:full_arb_solution}
In order to prove Theorem~\ref{full_arb_theorem} we will make use of the following lemma from Chan et al.~\cite{chan2011linear}: 
\begin{lemma}~\cite[Section 2.3,Theorem 6]{chan2011linear}
\label{CLSTW06b1_fix_theorem}
For any text $T$ of length $n$ characters over an alphabet of size $\sigma\leq n$ we can build an index of size $O(n\log n)$ bits so that given any pattern $p$ of length $m\geq \log^3 n\log\log n$ we can report all of the $occ$ substrings of the text which are at edit distance $1$ from $p$ in time $O(m+occ)$. 
\end{lemma}
In order to get Theorem~\ref{full_arb_theorem}, we combine this lemma with Theorem~\ref{arb_fix_theorem}. 
The combination is also straightforward. That is, we build both indices, where the index of Theorem~\ref{arb_fix_theorem} is built using the parameter $b=\log^3 n\log\log n$. Then if given a pattern of short length $m<\log^3 n\log\log n$, we use the index of Theorem~\ref{arb_fix_theorem} to answer in time $O((m+occ)/\epsilon)$, otherwise given a pattern of length $m\geq \log^3 n\log\log n$ we use the index of Lemma~\ref{CLSTW06b1_fix_theorem} to answer in time $O(m+\log^3 n\log\log n+occ)=O(m+occ)$. The space is dominated by our index, which uses $O(n\log n(b^\epsilon+\log^\epsilon n)/\epsilon)=O(n\log^{1+\epsilon}n/\epsilon)$ (by adjusting $\epsilon$).
\subsection{Tools}
\label{section:arb_tools}
In order to prove Theorem~\ref{arb_fix_theorem}, we will make use of the following additional tools:
\subsubsection{Heavy-light tree decomposition}
The heavy-light decomposition of a tree~\cite{HT84} decomposes the edges that connect a node to its children into two categories: heavy and light edges. Suppose that a given node $x$ in the tree has $w_x$ leaves in its subtree. We call $w_x$ the weight of $x$. We let $y$ be the heaviest child of $x$ ($y$ is the child of $x$ with the greatest number of leaves in its subtree). If there is more than one child sharing the heaviest weight, then we choose $y$ arbitrarily among them. Then $y$ is considered  a \emph{heavy} node and all the other children of $x$ are considered light nodes. The edge connecting $x$ to $y$ is considered a \emph{heavy} edge and all the other edges that connect $x$ to its children are considered light edges. A heavy-path is a maximal sequence of consecutive heavy edges. We will make use of the essential property summarized by  the following lemma:
\begin{lemma}
Any root to leaf path in a tree with $n$ nodes contains at most $\log n$ light edges.
\end{lemma}
Thus for a given pattern of length $m$ the traversal of a suffix tree decomposed according to a heavy path decomposition will contain at most $t=min(m,\log n)$ light edges. 
An example of a heavy-light decomposition of a suffix tree is illustrated in Figure~\ref{pic:heavy_path_suf_tree} in which the light edges are represented by dashed lines and the node weights are stored inside the nodes. 
\begin{figure}[htb] 
\centering\includegraphics[scale=1]{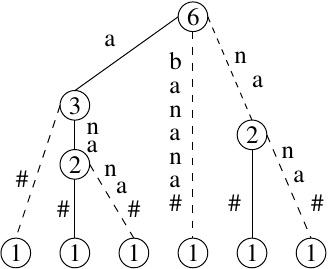} 
\caption{Heavy-light decomposition of a suffix tree for the string $banana$} 
\label{pic:heavy_path_suf_tree} 
\end{figure}

\subsubsection{1D colored range reporting data structure}
A 1D colored range reporting data structure solves the following problem: given an array $A[1..n]$ of colors each chosen from the same alphabet of size $\sigma$, answer to the following query: given an interval $[i,j]$ return all the $occ$ distinct colors which occur in the array elements $A[i],A[i+1]...A[j]$. A solution devised by Muthukrishnan~\cite{Mu02} uses $O(n\log n)$ bits of space and allows queries to be answered in optimal $O(occ)$ time. Later Sadakane~\cite{Sa07} improved the solution so as to use only $4n+o(n)+\sigma$ bits~\footnote{The term $4n+o(n)$ comes from the use of a succinct index for range minimum queries (RMQ). The term can be reduced to $2n+o(n)$ if the recent optimal solution of Fischer and Heun~\cite{Fi10} is used. The term $\sigma$ comes from the use of a bitvector of $\sigma$ bits that needs to be writable. The bit-vector is used to avoid reporting a single color more than once.} on top of the array $A$, reducing the space usage to $n\log\sigma+4n+\sigma+o(n)$ bits. 
\begin{lemma}~\cite{Sa07}
~\label{lemma:1D_color_rep}
Given an array $A[1..n]$ of colors from the set $\{1,2,,,\sigma\}$, we can build a data structure of size $O(n\log\sigma)$ so that given any range $[i..j]$ we can return all the $occ$ distinct colors which appear in $A[i]...A[j]$ in time $O(occ)$. Moreover the colors can be returned one by one in $O(1)$ time per color.
\end{lemma}

\subsection{High Level Description}
\label{section:arb_high_level_desc}

The solution of Theorem~\ref{constant_fix_theorem} has a too strong alphabet dependence in its query time. In particular, testing all insertion and substitution candidate strings takes $O(m\sigma)$ time as we have $O(m\sigma)$ candidates and spend $O(1)$ time for testing each candidate. By contrast, deletion candidates are at most $m$ and thus can be tested in $O(m)$ time. 
In order to reduce the number of candidates for insertions and substitutions we could use substitution lists as described in the solution of Belazzougui~\cite{B09}, which deals with approximate queries with one error over a dictionary of strings. 
A substitution list associates a list of characters $c_1,c_2\ldots$  to each pair of strings $(s_0,s_1)$ such that the string $s=s_0c_is_1$ is in the dictionary. When querying for substitutions on a pattern $q$ of length $m$, we can get the characters to substitute at position $i$ by querying the substitution list for the pair $(p[1..i-1],p[i+1,m])$. This will in effect return every character $c$ such that the string $s=p[1..i-1]cp[i+1,m]$ is in the dictionary. 

Concerning the space usage, the total space used by all the substitution lists is $O(n\log\sigma)$, where $n$ is the sum of lengths of the strings in the dictionary. This is because every string of length $t$ in the dictionary  contributes exactly $t$ entries to the substitution lists and the substitution lists are implemented in a succinct way that allows them to use $\log\sigma+O(1)$ bits per entry. 

We can now try to apply the idea of substitution lists to our problem. In our case we are indexing suffixes of a text of total length $O(n^2)$. A naive implementation of the substitution lists would use in total $O(n^2\log\sigma)$ bits of space. 
However it turns out that we only need to index strings of length at most $b=\operatorname{poly}(\log n)$ since we can use Lemma~\ref{CLSTW06b1_fix_theorem} to answer queries of length at least $b=\log^3 n\log\log n$ in time $O(m+occ)$. 

Now we could naively implement the idea by indexing all the factors of the text of length at most $b$. Because we have $O(nb)$ such factors with each of them being of length $b$, this would result in $O(nb^2\log\sigma)$ bits of space in total.

In order to further reduce the space, we will only store substitution lists for factors of length exactly $b$ rather than length at most $b$. This will reduce the space to $O(nb\log\sigma)$ bits. We will make use of a weak prefix search data structure that will allow us to get the contiguous range of substitution lists that correspond to a given prefix. The substitution list corresponding to that prefix will be obtained by merging all the substitution lists that are in the given range. 
However there could be duplication in the characters stored in each substitution list (the same character could be stored in multiple lists). In order to ensure that each character is only reported once we use Lemma~\ref{lemma:1D_color_rep}. 

We can further reduce the space of the substitution lists from $O(nb\log\sigma)$ to $O(n\log n\log\sigma)$ by making use of an observation that was made in Cole et al.~\cite{CGL04} and that was further used in Chan et al.~\cite{CLSTW10}. 

The main idea of Cole et al.~\cite{CGL04} was to build correction trees (which are of three types, deletion, substitution and insertion). The trees will store $O(n\log n)$ elements in total. The $\log n$ factor comes from the fact that each suffix may be stored up to $\log n$ times in the correction trees of a given type. More precisely a suffix is duplicated $O(t)$ times if the path from the root to that suffix has $t\leq \log n$ light edges. We will not make use of correction trees but will make essential use of the following related crucial observation: for each indexed factor we will need only to store $t\leq\log n$ elements in the substitution list. More precisely we will need only to store the characters that are labels of light edges. The reason is that a label of a heavy edge is already known when traversing the suffix tree and thus, need not be stored in a substitution list. Any suffix tree node has always one outgoing heavy edge and we can just consider its labeling character as a potential candidate for substitution or insertion. 

Finally we reduce the space further from $O(n\log n\log\sigma)$ to $O(n\log^{1+\epsilon}n)$ by using a more fine-grained encoding of the substitution lists. For that we will use a more precise categorization of the light edges according to their weight, and store for each node a list of all edges in each category. Then encoding a substitution list character can be done by using an index into the category of the edge labeled by that character. The crucial observation that allows for the use of less space lies in the fact that there cannot be too many heavier edges and, thus, an index into a category that contains heavier edges will use less space. 

\subsection{Data Structures for short patterns}
\label{section:DS_impl}

We will now present our solution that will be used to answer queries for patterns whose length is bounded by a certain parameter $b$. The implementation of our solution will consist of two parts. The first part reuses the same data structures used for Theorem~\ref{constant_fix_theorem} (see Section~\ref{section:const_fix_data_struct}), but with slight modifications. Our second part is an efficient implementation of substitution lists. We now describe our first part. We make only two modifications to the data structures of Section~\ref{section:const_fix_data_struct}: 
\begin{enumerate}
\item We use non compressed variants of the text indexing data structures. 
\item We mark the heavy edge of every suffix tree node. 
\end{enumerate}

The second part of our solution consists in a set of $b$ substitution list stores, where the substitution store number $i$ will contain all substitution lists for position $i$. 
Our substitution lists are built on the factors of length $b$ of the text (from now on we call them b-factors of the text). We thus have a dictionary of $n$ strings of length $b$ each. 
We start by building a compacted trie on those b-factors. This compacted trie is built by cutting the suffix tree below level $b$ (removing all nodes that have depth above $b$ and transforming internal nodes that span depth $b$ into leaves). 

We now give more detail on how the substitution stores are implemented. Recall that a substitution list as described in~\cite{B09} associates a list of characters $c_1,c_2\ldots$  to each pair of strings $(s_0,s_1)$ such that each string $s=s_0c_is_1$ is in the dictionary. In our case, the substitution lists are built by doing a top-down traversal of the compacted trie. Then each time we encounter a node associated with a factor $s_0$ and having a light edge labeled with character $c$, we do the following: we traverse all the leaves below that light edge (leaves of the subtree rooted at the target of that edge) and for each leaf associated with the string $s_0cs_1$, we append $c$ at the end of the substitution list corresponding to the pair $(s_0,s_1)$. 
We thus avoid storing a character associated with a b-factor when it is :
\begin{enumerate}
\item Part of a compacted path of an edge. 
\item A label of a heavy edge. 
\end{enumerate}
The order of the characters in a substitution list is arbitrary. We will have $b$ substitution stores, one per position $d\in[1..b]$. 
A substitution store for position $d$ will contain all substitution lists associated with pairs $(s_0,s_1)$ such that $|s_0|=d-1$. We first sort all those substitution lists by $s_0s_1$ (the concatenation of the two strings of their associated pairs). 
We let those sorted substitution lists be $L_1,L_2\ldots$. We store them in a compact, consecutive and packed way into an array $A_d=L_0L_1\ldots$. Note that this array uses space at most $n\log\sigma$ bits. This is because we have at most $n$ b-factors and we are storing at most one character from each b-factor (the character at position $d$). However when summed up over all the arrays, the total space (summing up $(|A_1|+\cdots+|A_b|)\log\sigma)$ will be $O(n\log n\log\sigma)$ and not $O(nb\log\sigma)$.
This is because each b-factor, i.e., root-to-leaf path in the trie, can contribute at most $\log n$ light edges. 

We now state how we implement a substitution store for position $d$. We will make use of the following data structures:
\begin{enumerate}
\item A constant-time weak prefix search on the set of strings $S_d$ formed by the concatenation of pairs of strings associated with substitution lists. That is for each pair $(s_0,s_1)$ we insert the string $s_0s_1$ in $S_d$. This weak prefix search will be able to return in constant-time the range of elements $s_0s_1$ prefixed by some string $p$ (the constant query time supposes some suitable preprocessing whose time is not taken into account). The space used by the weak prefix search data structure will be $O(nb^\epsilon)$ for any desired constant $0<\epsilon<1$. 
\item A 1D colored range reporting data structure implemented using Lemma~\ref{lemma:1D_color_rep} on top of the array $A_d$. This will allow to report the $occ$ distinct characters in any given range $A_d[i..j]$ in time $O(occ)$. The space used will be $O(|A_d|\log\sigma)$ bits. 
\item A prefix-sum data structure that stores the size of each of the substitution lists in sorted order. The space used will be $O(|A_d|)$ bits. In other words, given any $i$, the prefix-sum data structure will be able to return $\sum_{1\leq j\leq i}|L_i|$. 
\end{enumerate}

By extension, given a pair $(s_0,s_2)$ we will associate a substitution list formed by the union of the substitution lists associated with all pairs $(s_0,s_1)$ such that $s_1$ is prefixed by $s_2$. We call a substitution list $(s_0,s_2)$ a \emph{virtual} substitution list as such a list does not exist in the store, but is in reality obtained by the union of existing  substitution lists. The role of the weak prefix search and the 1D colored range reporting data structures is precisely to allow efficient computation of the virtual list $(s_0,s_2)$ as the union of all lists $(s_0,s_1)$. 

The $b$ weak prefix search data structures will be implemented using the same global hash function. 
\subsubsection{Query algorithm}

We now show how queries are implemented. We will concentrate on queries for substitution candidates. As pointed out before, queries for deletion candidates are handled in exactly the same way described in Section~\ref{section:const_query_algo1}. This is possible because we have $m$ candidates and we can exhaustively check each candidate within the same time bounds as in the constant alphabet size. Querying for insertion candidates is only slightly different from querying for substitution candidates and is omitted. 

Queries for substitution candidates are implemented in the following way: given a query string $q$ of length $m$, we traverse the trie top-down and find the locus of $q$. The locus of a string $q$ is defined by a pair $(x,\ell)$. We let $y$ denote the lowest node $y$ in the trie whose associated factor $p_y$ is prefix of $q$ (it is easy to see that such a node is unique). Then the locus $(x,\ell)$ of $q$ is defined in the following way: 
\begin{enumerate}
\item If $y$ has a child labeled with character $c=q[|p_y|+1]$, then $x$ will be the child of $y$ labeled with $c$ and $\ell$ will be the length of the longest common prefix between $q$ and $p_x$ where $p_x$ is the factor associated with $x$ (that is $p_x$ differs with $q$ on character number $\ell+1$, but agrees on characters $q[1..\ell]$). We call this a type $1$ locus. 
\item Otherwise ($y$ has no child labeled with character $c=q[|p_y|+1]$). Then $x=y$ and the locus will be $(y,|p_y|)$. We call this a type $2$ locus. 
\end{enumerate}
As we traverse the nodes from the root down to the node $x$, we identify all the nodes with their associated factors. Assume that the nodes are (in top-down order) $x_1,x_2,\ldots,x_t,x$ and their associated factors are $p_1,p_2,\ldots,p_t,p_x$ (note that the factors are of increasing length). Note that every factor $p_i$ is prefix of $q$ except $p_x$ which could potentially not be prefix of $q$. 

Notice that because of the suffix tree properties, we know that candidate positions for substitutions can only be $|p_1|+1,|p_2|+1,\ldots,|p_t|+1$ to which we add position $\ell+1$ in case $\ell<m$. To see why, notice that every prefix $p_r$ of $q$ of length between $|p_i|$ and $|p_{i+1}|$ (or of length less than $|p_1|$ or  between $|p_t|$ and $|p_x|$) is not associated with any suffix tree node, but still appears as a factor in the text. This means that there is a unique character $c$ such that $p_rc$ appears in the text (if there were two such characters then there would have been a suffix tree node associated with $p_r$). However $p_rc$ is also prefix of $q$ and thus agrees on character $c$ with $q$. Thus we have no character that we can substitute for character $c$ in $q$ and yet obtain a factor appearing in the text. 

In summary, the substitution candidate characters will be :
\begin{enumerate}
\item For every node $z=x_i$ with $i\in[1..t]$ (and also for the node $z=x$ when the locus is of type $2$) we have the following candidates:
\begin{enumerate}
\item The character that labels a heavy child of $z$ at position $|p_z|+1$ (recall that $p_z$ is the factor associated with $z$) if this character is different from $c=q[|p_z|+1]$. 
\item The characters obtained from the substitution list corresponding to the pair $(p_i,q[|p_i|+2..m])$ stored in the substitution store number $|p_z|+1$ (we naturally avoid querying for the character $c=q[|p_z|+1]$ if it is contained in the substitution list). We only test one of those characters by substituting it at position $|p_z|+1$ and, if we have a match, then we can report all the remaining characters in the list (except character $q[|p_z|+1]$ if it appears in the list). Details on how the substitution store is queried are below. 
\end{enumerate}
\item In case of a type $1$ locus and $\ell<m$, we also try to substitute the character $p_x[\ell+1]$ at position $\ell+1$. 
\end{enumerate}

We now give more details on how the substitution store is queried. Recall that a substitution store number $d$ stores the substitution lists associated with the pairs $(s_0,s_1)$ such that $|s_0|=d-1$ and $|s_0|+1+|s_1|=b$ and moreover there exists at least one character $c$ such that $s_0cs_1$ is a b-factor of the text (the substitution list contains at least the character $c$). 

Thus in case we have $m=b$, we can directly query the relevant substitution list. 
Otherwise, we first query the weak prefix search data structure which returns a range of substitution lists, then use the prefix-sum data structure to convert this range into a range in the array $A_d$.

We finally use the 1D-colored rangeed reporting data structure to report all the distinct colors in the range. We need to use the special colored range reporting data structure because the colors may be duplicated in different substitution lists. We only test for the first character (the first reported color), and if it gives a match, then we go on and substitute the following characters (the following colors obtained from the colored rangeed range reporting data structure). 


In the next subsection, we describe how to reduce the space used per substitution list element from $O((\log\sigma+\log^\epsilon n/\epsilon))$ to just $O(\log^\epsilon n/\epsilon)$, hence removing any dependence on $\sigma$.

\subsection{Space reduction using generalized weight decomposition}
\label{section:gen_weight_decomp}
In order to reduce the space we will use a more fine-grained categorization of the nodes. Instead of having only two types of nodes (heavy and light), we will instead define multiple types of nodes according to their weight. We first define the generalized weight decomposition, then prove an important lemma on the decomposition and finally show how to use it to reduce the space usage of our index. 
\subsubsection{Generalized weight decomposition}
We order the children of a node $x$ by their weights (from the heaviest to the lightest node). We say that a child $u'$ has rank $i$ if it is the $i$th heaviest child of its parent. By extension we say that the edge connecting $x$ to $y$ is of rank $i$. 
We have the following lemma:
\begin{lemma}
\label{lemma:heav_tree_lemma}
Any root-to-leaf path in a tree with $n$ nodes contains at most $\log_i n$ light edges of rank at least $i$.
\end{lemma}
\begin{proof}
The proof is easy. We first prove that traversing an edge of rank $i$ reduces the weight of the current node by a factor at least $i$. This is easy to see by contradiction. Suppose $x$ has weight $w_x$ and its child $y$ of rank $i$ has weight $w_y>(w_x/i)$. Then necessarily all the $i$ children of rank at most $i$ are at least as heavy as $y$ and thus have necessarily weight at least $w_y$ each. 
Thus the total weight of the first $i$ children of $x$ is at least $iw_y>w_x$ which is more than the weight of $x$. 
We thus have proved that traversing an edge of rank $i$ reduces the weight of the current node by a factor at least $i$. 
From there we can easily prove our lemma. We start the traversal at the root which has weight $n$ and each time we traverse a child of rank at least $i$ we divide the weight by a factor at least $i$. After traversing $\log_i n$ nodes we reach a node with weight $\frac{n}{\log_in}=1$ which by definition can not contain any light edge in its subtree. Thus we reach a leaf after traversing at most $\log_i n$ nodes of rank at least $i$. 
\qed
\end{proof}
An observation similar to the lemma above was already made in a recent paper by Bille et al.~\cite{BGSV12}. 

\subsubsection{Implementation}
The generalization works in the following way: we categorize the light children of a node into at most $\lceil1/\epsilon\rceil$ categories. Each category $i$ contains the children whose rank is between $2^{\log^{(i-1)\epsilon} n}$ and $2^{\log^{i\epsilon} n}-1$. 
To implement our generalized method, we reuse exactly the same data structures used in Section~\ref{section:DS_impl}, except for the following two points:
\begin{enumerate}
\item In the suffix tree, we order all the labels of the children of every suffix tree node by decreasing weight. For a given node $x$, with $t$ children, we order the children by decreasing weight and store an array $C[1..t]$ where we put at position $i$ the character that labels the child number $i$ in the ordering. 
\item Then we will use $b\cdot \lceil 1/\epsilon\rceil$ substitution stores instead of $b$ substitution stores. More precisely we store a substitution store for each position $i\in[1..b]$ and category $j\in[1..\lceil1/\epsilon\rceil]$. 
\end{enumerate}
The $\lceil 1/\epsilon\rceil$ substitution stores for the same position will all share the same weak prefix search data structure (that will store exactly the same content that was originally stored when there was a single substitution store for the position), but use $\lceil 1/\epsilon\rceil$ different $1D$ colored range reporting and prefix-sum data structures. Note that the prefix-sum data structures could possibly indicate empty substitution lists, in case there exists another non-empty substitution list with the same pair $(s_0,s_1)$. 
We now bound the space incurred by the modifications. The table $C$ uses $t\log\sigma$  bits and, when added over all the suffix tree nodes, the space becomes $O(n\log\sigma)$ bits. We will prove that all the substitution stores with the same category $j$ will use in total $O(n\log^{1+\epsilon} n)$ bits of space. 
In order to reduce the space used by the substitution lists (stored using the $1D$ colored range reporting and the prefix-sum data structures), we will use variable lengths to code the characters. More precisely, instead of storing the character (in the $1D$ colored range reporting data structure) using $\log\sigma$ bits, we will store the rank of the child (among the children of the parent of that child) labeled by that character. 

Then using the table $C$, we can easily recover the character. The fact that category $j$ contains nodes of rank at most $2^{\log^{j\epsilon} n}-1$, means that we only need $\log (2^{\log^{j\epsilon} n}-1)\leq \log^{j\epsilon} n$ bits to encode each substitution list element in a substitution store of category $j$. 

We note that each root-to-leaf path has at most $(\log n)^{1-(j-1)\epsilon}$ nodes of category $j$. This is the case as nodes in category $j$ have rank at least $r=2^{(\log n)^{(j-1)\epsilon}}$. Thus by Lemma~\ref{lemma:heav_tree_lemma} any root-to-leaf path has at most $\log_r n=\frac{\log n}{\log r}=\frac{\log n}{(\log n)^{(j-1)\epsilon}}=(\log n)^{1-(j-1)\epsilon}$ light edges of rank at least $r$, and we deduce that the substitution stores of category $j$ will have at most $O(n(\log n)^{1-(j-1)\epsilon})$ elements. As each element is encoded using $\log^{j\epsilon} n$ bits, we conclude that the total space to encode the substitution lists of category $j$ is $O(\log^{j\epsilon} n\cdot \log^{1-(j-1)\epsilon}n)=O(n\log^{1+\epsilon}n)$ bits of space.

When multiplied over all the (at most) $\lceil1/\epsilon\rceil$ categories, we get total space $O(n\log^{1+\epsilon}n/\epsilon)$ bits.
This bounds the space used by the $1D$ colored range reporting data structure. The space used by the prefix-sum data structures will be at most $O(n\log n/\epsilon)$, since for every non-empty substitution list associated with a pair $(s_0,s_1)$, we could have up to $\lceil 1/\epsilon\rceil$-1 empty ones. 
The space used by the weak prefix search data structure remains the same. 
Overall, the total space will be $O(n\log^{1+\epsilon}n/\epsilon)$ bits. 
\subsubsection{Queries}
We only describe queries for substitutions. A query will work in the following way: we traverse the suffix tree top-down and for each node of the suffix tree at depth $d$, we use the label of the heavy child as a candidate character for substitution (or insertion). 
To get candidate characters that label light children, we query the substitutions stores $(d,j)$ for every $d\in[1..\lceil1/\epsilon\rceil]$ separately. We first start by querying the weak prefix search data structure (which is shared by all the 
$\lceil1/\epsilon\rceil$ substitution stores). We then query successively all the $\lceil1/\epsilon\rceil$ prefix-sum data structures associated with the same pair $(s_0,s_1)$ until we find the first non-empty one, in which case we retrieve the characters in the substitution list from the $1D$ colored range reporting structure and check that the first character gives a match (exactly as in Section~\ref{section:DS_impl}). If that is the case, we continue reporting all the other characters in the substitution list and in all the other non-empty substitution lists associated with the pair $(s_0,s_1)$. If the first substituted character did not give a match, then we know that all the other substitution lists will also  not give a match. 

Since we have in total up to $\lceil m/\epsilon\rceil-1$ empty substitution stores, and a query on each substitution store takes constant time independent of $\epsilon$~\footnote{Note that a query on all but the first non-empty substitution store takes constant time per reported character, since it involves querying the prefix-sum and the $1D$ colored range reporting data structures which answer in, respectively, constant time and constant time per element.}, the new query algorithm will incur an additive $O(m/\epsilon)$ term in the query time. 

\subsection{Handling very large $\sigma$}
Throughout the paper, we have assumed that $\sigma\leq n$. In this section we show how to handle the case $\sigma\geq n$. To handle this case, we can use an indexable dictionary $D$ that stores the set $\Sigma_T$ of characters which appear in the text $T$. As a byproduct, the indexable dictionary will associate a unique number $r_\alpha\in[1..|\Sigma_T|]$ to each character $\alpha\in \Sigma_T$. Moreover, the dictionary should also be able to do the reverse mapping, returning $\alpha$ when given $r_\alpha$. 
In order to implement the indexable dictionary, we can use the data structure of~\cite{RRS07} which supports the mapping and reverse mapping in constant time and uses $O(|\Sigma_T|\log\sigma)\leq O(n\log\sigma)$ bits of space which is less than the space used by the index of Theorem~\ref{full_arb_theorem}. 
Then, we replace the character $\alpha$ in the text $T$ with $r_\alpha$, giving a new text $T'$. Then the size of the alphabet of $T'$ will be $|\Sigma_T|\leq n$ and we can build the index of Theorem~\ref{full_arb_theorem} on $T'$ rather than $T$. 

To implement a query on a string $q[1..m]$, we first query the dictionary $D$ for each character $q[i]$ obtaining the number $r_{q[i]}$. We now have three cases:
\begin{enumerate}
\item If all characters of $q$ have been found in $D$, then we replace each $q[i]$ with $r_{q[i]}$ obtaining a new string $q'$ and query the index of Theorem~\ref{full_arb_theorem}. 
Finally when querying the index, we should use the dictionary $D$ to translate the characters output by the index to the original alphabet (that is given $r_\alpha$, return $\alpha$). 
\item If exactly one character of $q$ was not found in the dictionary, we conclude that the only two possible errors are the deletion or substitution of that character. We thus can just query the suffix tree for an exact matches on the two strings obtained by deleting or substituting the character. 
\item If two or more of the characters of $q$ were not found in $D$, then this means we have more than one error and the result of the query will be empty. 
\end{enumerate}


\section{Conclusion}
\label{section:conclusion}
In this paper we have presented new and improved solutions for indexing a text $T$ for substring matching queries with one edit error. 
Given a query string $q$ of length $m$ our indices are able to find the $occ$ occurrences of substrings of $T$ at edit distance $1$ from $q$ in time $O(m+occ)$. Our first index is only suitable for constant alphabet sizes and uses $O(n\log^\epsilon n)$  bits of space, where $n$ is the length of the text $T$ and $\epsilon$ is any constant such that $0<\epsilon<1$. The same index can be tuned to achieve a different tradeoff, namely to use $O(n\log\log n)$ bits of space while answering to queries in time $O((m+occ)\log\log n)$. 

The other index presented in this paper is suitable for an arbitrary alphabet size and uses $O(n\log^{1+\epsilon} n)$ bits of space while answering to queries in optimal time $O(m+occ)$. 

It is an interesting open problem to show whether our indices are optimal or whether they can be improved. It is interesting to compare our indices with known results for substring matching with zero errors. For constant-sized alphabets, the best known results achieve  optimal $O(m/\log n+occ)$ query time using $O(n\log^\epsilon n)$ bits of space~\cite{GV05}. The $m/\log n$ is known to be the best possible for constant-sized alphabets as the time needed to read the query string is at least $\Omega(m/\log n)$ (assuming that the characters of the strings are tightly packed). Our first index achieves the same space usage but has a $\log n$ factor slowdown in the length of the query string. It is an interesting open problem to show whether this $\Omega(\log n)$ factor is inherent in the substring matching with one error when the space is required to be small. 

For arbitrary alphabet sizes, the $O(m+occ)$ query time of our second solution is optimal as it is known that $\Omega(m)$ time is needed in the worst case to read the query string when the alphabet is of arbitrary size~\footnote{Each character could occupy up to $\log n$ bits which means that we need at least $\Omega(1)$ time to read each character in a our RAM model with $w=\Theta(\log n)$.}. On the other hand the space usage $O(n\log^{1+\epsilon} n)$ is a factor $\log^\epsilon n$ more than the $O(n\log n)$ bits used by standard indices for substring matching with zero errors (suffix arrays and suffix trees). It is an intriguing open problem to show whether this $\log^\epsilon n$ factor is inherent in substring matching with one error or whether it can be removed. 

\section*{Acknowledgements}
The author wishes to thank the anonymous reviewers for their helpful comments and corrections
and Travis and Meg Gagie for their many helpful corrections and suggestions.

\small 
\bibliography{full_text_edit_dist} 
\normalsize

\newpage

\end{document}